\documentclass[oneside,english]{amsart}
\usepackage[T1]{fontenc}
\usepackage[latin9]{inputenc}
\usepackage{color}
\usepackage{float}
\usepackage{amsthm}
\usepackage{amssymb}

\makeatletter

\providecommand{\tabularnewline}{\\}
\floatstyle{ruled}
\newfloat{algorithm}{tbp}{loa}
\providecommand{\algorithmname}{Algorithm}
\floatname{algorithm}{\protect\algorithmname}

\numberwithin{equation}{section}
\numberwithin{figure}{section}
\theoremstyle{plain}
\newtheorem{thm}{\protect\theoremname}[section]
  \theoremstyle{definition}
  \newtheorem{defn}[thm]{\protect\definitionname}
  \theoremstyle{definition}
  \newtheorem{example}[thm]{\protect\examplename}
  \theoremstyle{plain}
  \newtheorem{prop}[thm]{\protect\propositionname}
  \theoremstyle{remark}
  \newtheorem{rem}[thm]{\protect\remarkname}

\makeatother

\usepackage{babel}
  \providecommand{\definitionname}{Definition}
  \providecommand{\examplename}{Example}
  \providecommand{\propositionname}{Proposition}
  \providecommand{\remarkname}{Remark}
\providecommand{\theoremname}{Theorem}

\begin{document}
\medskip{}

\title{Fully Polynomial Time Approximation Schemes (FPTAS) for some counting problems}

\author{Tzvi Alon\medskip{}
}

\maketitle
\textsc{This thesis was carried out under the supervision of dr. Nir Halman}

\medskip{}

\pagebreak{}
\begin{abstract}
In this thesis we develop FPTASs for the counting problems of $m-$tuples, contingency tables with two rows, and
0/1 knapsack. For the problem of counting $m-$tuples, we design two algorithms, one is strongly polynomial. As
far as we know, these are the first FPTASs for this problem. For the problem of counting contingency tables we improve
significantly over the running time of existing algorithms. For the problem of counting 0/1 knapsack solutions,
we design a simple strongly polynomial algorithm, with similar running times to the existing algorithms.

Our results are derived by using, as well as expanding, the method of $K-$approximation sets and functions introduced
in \cite{key-5}.
\end{abstract}

\section*{Acknowledgments}

I would like to thank my advisor, Nir Halman, for his generous time and commitment. His guidance and endless patience
helped me direct myself to a fruitful thinking, and I am grateful for that.

I wish to thank my dear parents, my wife Aderet for her support and understanding, and to my sweet children Tuvya,
Tzuri and Tziyon.

\section{Introduction}

\subsubsection{Organization of this thesis:}

In Section \ref{sub:The-problems} we define the problems under consideration. These problems can be solved via
dynamic programming, a technique which we survey in Section \ref{sub:Dynamic-programming}. In Section \ref{sub:Approximation-algorithms}
we survey approximation algorithms in general, and in Section \ref{sub:approximation-sets-and} we survey the method
of $K-$approximation sets and function, a specific method used to derive approximation algorithms. In Sections
\ref{sec:Counting-tuples}, \ref{sec:contingency-tables-with} , \ref{sec:0/1-knapsack--strongly} and \ref{sec:Strongly-polynomial-m-tuples}
we develop the approximation algorithms for our problems.

\subsubsection{Notations:}
\begin{itemize}
\item Denote by$\mathbb{N}$ the set of natural numbers, i.e. $\mathbb{N}=\left\{ 1,2,3,\ldots\right\} $.
\item Let $\mathbb{Z}^{+}$ be $\mathbb{N}\bigcup\left\{ 0\right\} $.
\item Let $\log z$ be the base 2 logarithm of $z$.
\item For a nondecreasing function $\varphi:\left\{ A,\ldots,B\right\} \rightarrow\mathbb{R}$ let $\varphi^{\mbox{max}}=\varphi\left(B\right)$
\end{itemize}

\subsection{The problems\label{sub:The-problems}}

For every optimization problem (or decision problem) there is a corresponding counting problem, in which we count
the number of feasible solutions (it is important to emphasize that we do not count the number of \emph{optimal}
solutions, but the number of \emph{feasible} solutions). 

In this thesis we deal with the following three problems, which are known to belong to the class of \#P-hard problems-
a class of counting problems which cannot be solved in a polynomial time unless P=NP (\cite{key-7} p. 167).\\

\textbf{$m-$tuples (\cite{key-7} p. 225 {[}SP21{]}):}

\emph{Input: }Sets $X_{1},\ldots,X_{m}\subseteq\mathbb{Z}^{+}$, $K,B\in\mathbb{N}$. For all $1\leq i\leq m$ and
$1\leq j\leq\left|X_{i}\right|$ denote by $x_{ij}$ the $j$'th element of the set $X_{i}$.

\emph{Decision problem:} Are there $K$ or more $m-$tuples $\left(x_{1\ell_{1}},...,x_{m\ell_{m}}\right)$ for
which $\sum_{k=1}^{m}x_{k\ell_{k}}\geq B$ ?

\emph{Counting problem: } How many $m-$tuples $\left(x_{1\ell_{1}},...,x_{m\ell_{m}}\right)$ are there such that
$\sum_{k=1}^{m}x_{k\ell_{k}}\geq B$ ?

The problem belongs to the class of \#P-hard problems (\cite{key-7} p. 225 {[}SP21{]}).

The input for the counting problem does not include $K$, therefore the input size for this problem is $O\left(\log B+\sum_{i=1}^{m}\sum_{j=1}^{\left|X_{i}\right|}\log x_{ij}\right)$.\\

\textbf{Contingency tables(\cite{key-2}):}

\emph{Input:} $r=\left(r_{1},...,r_{m}\right),\, s=\left(s_{1},...,s_{n}\right)$,$N$, all belongs to $\mathbb{N}$,
such that $r,s$ are partitions of $N$.

\emph{Counting problem : }The set $\Sigma_{s,r}$ of contingency tables with row sums $s$ and column sums $r$
is defined by 
\[
\Sigma_{r,s}=\left\{ Z\in\mathbb{Z}^{+^{m\times n}}:\sum_{j=1}^{n}Z_{ij}=r_{i}\mbox{ for }1\leq i\leq m,\sum_{i=1}^{m}Z_{ij}=s_{j}\mbox{ for \ensuremath{1\leq j\leq n}}\right\} 
\]
what is the cardinality of $\Sigma_{r,s}$?

The problem belongs to the class of \#P-hard problems, see Theorem 1 in \cite{key-8}.

The input for the counting problem does not include $K$, therefore the input size for this problem is $O\left(\log N+\sum_{i=1}^{m}\log r_{i}+\sum_{i=1}^{n}\log s_{i}\right)$.\\

\textbf{0/1 knapsack (\cite{key-7} p. 247 {[}MP9{]}):}

\emph{Input:} $w_{1},\ldots,w_{n}\in\mathbb{N}$ weights of $n$ items, $v_{1},\ldots,v_{n}$ the values of this
items, $C\in\mathbb{N}$ the capacity of the knapsack, and $K\in\mathbb{N}$.

\emph{Decision problem:} Is there a subset $S\subseteq\left\{ 1,\ldots,n\right\} $ such that $\sum_{s\in S}w_{s}\leq C$
and such that $\sum_{s\in S}v_{s}\geq K$?

\emph{Counting problem (\cite{key-3})}\textbf{: }What is the cardinality of $s_{n}\left(C\right):=\#\left\{ S\subseteq\left\{ 1,...,n\right\} |\sum_{k\in S}w_{k}\leq C\right\} $?

The problem belongs to the class of \#P-hard problems (Section 1 in \cite{key-10}).

The input for the counting problem does not include $v_{1},\ldots,v_{n},K$, therefore the input size for this problem
is $O\left(\log C+\sum_{i=1}^{n}\log w_{i}\right)$.

\subsection{Dynamic programming (DP)\label{sub:Dynamic-programming}}

Dynamic programming is a method used for solving a complex problems which can be broken down into a collection of
simpler subproblems, sharing a similar structure as the main problem. Each of these subproblems is solved only once,
and in the next time the same subproblem occurs, instead of recomputing it, one can just use the solution already
computed. Hopefully, this way we can reduce the amount of memory and computations needed to solve the problem\cite{key-9}.

To demonstrate DP, we show now in detail the DP formulation for counting 0/1 knapsack. Let $w_{1},..,w_{n},C$ be
an instance of 0/1 knapsack. Let 
\begin{equation}
s_{i}\left(j\right)=\#\left\{ S\subseteq\left\{ 1,...,i\right\} |\sum_{k\in S}w_{k}\leq j\right\} \label{eq:knapsack set}
\end{equation}
We want to evaluate $s_{n}\left(C\right)$. We first consider the boundary conditions: For $j<0$, we have $s_{i}\left(j\right)=0$
for every $i$. For $i=0$ we have $s_{0}\left(j\right)=1$ for every $j$ (we fill the knapsack with no items,
so there is only the empty solution).

For $i=1,\ldots,n$ the following recursion is valid:

\begin{equation}
s_{i}\left(j\right)=s_{i-1}\left(j\right)+s_{i-1}\left(j-w_{i}\right)\label{eq:knapsack dynamic}
\end{equation}

\emph{Explanation:} We consider two cases:

Case 1: the $i$'th item is in the knapsack. The remaining capacity in the knapsack for items $1,\ldots,i-1$ is
therefore $j-w_{i}$, so there are $s_{i-1}\left(j-w_{i}\right)$ solutions.

Case 2: the $i$'th item is not in the knapsack. Then the capacity for items $1,\ldots,i-1$ is $j$, and for this
there are $s_{i-1}\left(j\right)$ solutions.

We aim to calculate the value $s_{n}\left(C\right)$. For this, we start calculating $s_{0}\left(\cdot\right)$,
continue with $s_{1}\left(\cdot\right)$ by recursion (\ref{eq:knapsack dynamic}), and so on until we get to $s_{n}\left(\cdot\right)$.

How many operations are done in this calculation? For every iteration $i$ we need to calculate the $C$ values
$\left\{ s_{i}\left(1\right),\ldots,s_{i}\left(C\right)\right\} $. We have $n$ iterations, so the number of values
we calculate sum up to $nC$. Since in each calculation of a value we use $O\left(1\right)$ operations, the running
time to compute $s_{n}\left(C\right)$ is $O\left(nC\right)$.

Now, at first sight this running time looks polynomial in the input size. But, since numbers are stored in the computer
in binary encoding, the number of bits used to store $C$ is $O\left(\log C\right)$. Therefore, the running time
of the above algorithm is in fact \emph{exponential} in the input size.

The running time of $O\left(nC\right)$ is called \emph{pseudo-polynomial}. I.e. it is polynomial in the numbers
of the problem, but exponential in the input size. 

Recall that the problem is \#P-hard, so unless P=NP, not only a DP algorithm must be intractable, but also any other
exact algorithm. One way to tackle this hardness result is to get an approximate solution in polynomial time.

Before approaching to deal with approximation algorithms, let us define the notion of strongly polynomial algorithm.
An algorithm is s\emph{trongly polynomial} if the number of elementary steps is polynomially bounded in the dimension
of the input, where the \emph{dimension of the input} is the number of data items in the input (that is, each number
is considered to add one to the dimension of the input)\cite{key-12}. I.e. if an algorithm is strongly polynomial,
then the number of elementary steps is independent of the numbers in the input. (The original definition contains
another condition about rational numbers, which is not relevant here, and is therefore omitted).

\subsection{Approximation algorithms\label{sub:Approximation-algorithms}}

Since the class of \#P-hard is believed to be intractable, we turn to polynomial time approximation algorithms.

There are several measures of approximation, such as relative error or additive error. Here we deal only with relative
error approximations (the rest of this section is based on \cite{key-1} p. 86-116).

Let $P$ be a problem, $x$ be an instance, and $P\left(x\right)\in\mathbb{R}^{+}$ be the solution value. 
\begin{defn}
\label{def:K-apx-algo}Let $K>1$. We say that an algorithm is a $K-$ \emph{approximation algorithm} (or \emph{constant
factor approximation algorithm}) of problem $P$ if for every instance $x$ returns $s\left(x\right)$ such that
$P\left(x\right)\leq s\left(x\right)\leq KP\left(x\right)$.
\end{defn}
When a problem is intractable, we would like to find $K-$approximations for smaller and smaller values of $K$,
that bring us as close as possible to the solution. Of course, we will pay for the greater accuracy in larger running
time.
\begin{defn}
We say that an algorithm is a Polynomial Time Approximation Scheme (PTAS) if for any given $\epsilon>0$ it is a
$\left(1+\epsilon\right)-$approximation algorithm that runs in time polynomial in the input size.
\end{defn}
While polynomial in the input size, a PTAS may be exponential in $\frac{1}{\epsilon}$, e.g. $O\left(n^{2}2^{\frac{1}{\epsilon}}\right)$.
I.e. the dependency on the quality of the approximation may be very large. In fact, this dependence may sometimes
prevent any practical use of the scheme. A better situation is when the algorithm's running time is polynomial also
in the approximation ratio:
\begin{defn}
We say that an algorithm is a Fully Polynomial Time Approximation Scheme (FPTAS) if for any given $\epsilon>0$
it is a $\left(1+\epsilon\right)-$approximation algorithm that runs in time polynomial in both $\frac{1}{\epsilon}$
and the input size, e.g. $O\left(\frac{n^{2}}{\epsilon}\right)$. 
\end{defn}
In this paper, we develop an FPTAS for the problems mentioned above via the method of $K-$approximation sets and
functions which we survey in the next section.

\subsection{$K-$approximation sets and functions\label{sub:approximation-sets-and}}

In this section we survey the method of $K-$approximation sets and functions as defined in \cite{key-5}.

The next definition of $K-$approximation functions is similar to Definition \ref{def:K-apx-algo}:
\begin{defn}
Let $f:S\rightarrow\mathbb{R}^{+}$ be an arbitrary function, and $K\geq1$. We say that $\tilde{f}$ is a $K-$\emph{approximation
function }of $f$ if $f\left(x\right)\leq\hat{f}\left(x\right)\leq Kf\left(x\right)$ holds for every $x\in S$.
\end{defn}
We start the discussion about the method of $K-$approximation sets and functions with an example.
\begin{example}
\label{exa:}Let $\varphi:\left\{ 1,\ldots,M\right\} \rightarrow\mathbb{Z}^{+}$ be the identity function, i.e.
$\varphi\left(i\right)\equiv i,\,\forall i$. Suppose we want to store $\varphi$ on the computer memory as tuples
$\left(i,\varphi\left(i\right)\right)$. If $M$ is a small number, we should not have any problem. But if $M$
is a big number, we may not have enough space to store all such tuples. However, we can be satisfied with a $2-$approximation
of the function: For every $1\leq i\leq M$ define $\tilde{\varphi}\left(i\right)=2^{j}$ if $2^{j-1}<i\leq2^{j}$,
i.e. $\tilde{\varphi}\left(i\right)=2^{\lceil\log i\rceil}$. It is easy to see that for every $1\leq i\leq M$
we have $\varphi\left(i\right)\leq\tilde{\varphi}\left(i\right)\leq2\varphi\left(i\right)$, thus $\tilde{\varphi}$
is a $2-$approximation function of $\varphi$. Now, note that the problem of not having enough space is solved:
we can store only the set $\left\{ \left(2^{i},\tilde{\varphi}\left(i\right)\right)|i=0,\ldots,\lfloor\log M\rfloor\right\} \bigcup\left\{ \left(M,\tilde{\varphi}\left(M\right)\right)\right\} $.
The cardinality of this set is $O\left(\log M\right)$. Extracting $\tilde{\varphi}\left(i\right)$ for an arbitrary
$1\leq i\leq M$ can be done by binary search on the set $\left\{ 1,2,\ldots,2^{\lfloor\log M\rfloor},M\right\} $
in $O\left(\log\log M\right)$ time. The principles used to overcome the space problem is to pick wisely representatives
of the original function $\varphi$, and construct by them an approximation function. $\square$
\end{example}
In the next paragraphs we consider the arguments of Example \ref{exa:} for an arbitrary monotone function $\varphi$.
To simplify the discussion, from now on we modify Halman et al.'s definitions in \cite{key-5} to integer-valued
nondecreasing functions over intervals of integer numbers.

Let $D=\left\{ A,A+1,\ldots,B|A,B\in\mathbb{Z}\right\} $ be a finite interval, and $\varphi:D\rightarrow\mathbb{Z}^{+}$
be an arbitrary nondecreasing function over $D$. Suppose $\varphi$ is accessed via an oracle in $t_{\varphi}$
time units. The input for this problem is $A,B,\varphi^{max}=\varphi\left(B\right)$, so the input length is $O\left(\log A+\log B+\log\varphi\left(B\right)\right)$.
Of course, by querying all values $\varphi\left(x\right)$ in $x\in D$ and storing them in a sorted array of the
form $\left\{ \left(x,\varphi\left(x\right)\right)|x\in D\right\} $, we can obtain in $O\left(\left|B-A\right|t_{\varphi}\right)$
time a representation of size $O\left(\left|B-A\right|\right)$ which can return the value $\varphi\left(x\right)$
for any $x$ in $O\left(\log\left|B-A\right|\right)$ time. However, $\left|B-A\right|$ is not necessarily polynomially
bounded in the input size. 

The method of $K-$approximation sets and functions enables us to build an approximation that is both \emph{succinct}
(of size polylogarithmic in the input size) and \emph{efficient} (can be built in time polylogarithmic in the input
size).
\begin{defn}
\label{def:K apx set}Let $K\geq1$, and let $\varphi:\left\{ A,...,B\right\} \rightarrow\mathbb{Z}^{+}$ be a nondecreasing
function. Let $W=\left\{ w_{1},...,w_{r}\right\} $ be a subset of $\left\{ A,...,B\right\} $, where $A=w_{1}<w_{2}<...<w_{r}=B$. \end{defn}
\begin{itemize}
\item We say that $W$ is a \emph{$K-$approximation set} of $\varphi$ if $\varphi\left(w_{j+1}\right)\leq K\varphi\left(w_{j}\right)$
for each $j=1,...,r-1$ that satisfies $w_{j+1}-w_{j}>1$.
\item The \emph{approximation of $\varphi$ induced by $W$ }is:
\[
\hat{\varphi}\left(x\right)=\begin{cases}
\varphi\left(x\right) & x\in W\\
\varphi\left(w_{i+1}\right) & w_{i}<x<w_{i+1}\mbox{ for some \ensuremath{i}}
\end{cases}
\]

\end{itemize}
The following 2 propositions show the usefulness of the arguments in Definition \ref{def:K apx set} to achieve
a succinct $K-$approximation function. Proposition \ref{prop:4.5} tells us that the approximation function induced
by a $K-$approximation set is indeed a $K-$approximation function, and Proposition \ref{prop:4.6} is about how
we can construct efficiently a succinct $K-$approximation set.
\begin{prop}
\label{prop:4.5}(Based on Proposition 4.5 of \cite{key-5}):Let $\varphi:\left\{ A,...,B\right\} \rightarrow\mathbb{Z}^{+}$
be a nondecreasing function. Let $K\geq1$, and let $W$ be a $K-$approximation set of $\varphi$. Let $\hat{\varphi}$
the approximation of $\varphi$ induced by $W$. Then $\hat{\varphi}$ is a nondecreasing $K-$approximation function
of $\varphi$. In addition, if $\varphi$ is stored as a sorted array $\left\{ \left(x,\varphi\left(x\right)\right)|x\in W\right\} $,
then for any $x\in\left\{ A,\ldots,B\right\} $, $\hat{\varphi}\left(x\right)$ can be determined in $O\left(\log\left|W\right|\right)$
time.
\end{prop}
In Algorithm 1 of \cite{key-5}, Halman et al. introduce the algorithm \textsc{ApxSet}$\left(\varphi,D,x^{*},K\right)$
which when given as arguments (i) a unimodal discrete function $\varphi$ with (ii) a finite domain $D$ of real
numbers which is (iii) minimized at $x^{*}$, and (iv) an approximation ratio $K$, it returns a $K-$approximation
set for $\varphi$. Here we use this algorithm for nondecreasing functions with interval in $\mathbb{Z}$ as a domain.
Thus, for simplicity, we omit $x^{*}$ from the input, and denote the algorithm by \textsc{ApxSet}$\left(\varphi,\left\{ A,\ldots,B\right\} ,K\right)$. 

\begin{algorithm}[h]
\protect\caption{\label{alg:ApxSet}Constructing a $K-$approximation set for a nondecreasing function $\varphi$}

\begin{enumerate}
\item \textbf{Function }\textbf{\textsc{ApxSet}}$\left(\varphi,\left\{ A,\ldots,B\right\} ,K\right)$
\item $x\leftarrow B$
\item $W\leftarrow\left\{ A,B\right\} $
\item while $x>A$ do

\begin{enumerate}
\item $x\leftarrow\min\left\{ x-1,\min\left\{ y\in\left\{ A,\ldots,B\right\} |K\varphi\left(y\right)\geq\varphi\left(x\right)\right\} \right\} $
\item $W\leftarrow W\bigcup\left\{ x\right\} $
\end{enumerate}
\item end while
\item \textbf{return} $W$\end{enumerate}
\end{algorithm}

\begin{prop}
\label{prop:4.6}(Based on Proposition 4.6 in \cite{key-5})Let $D=\left\{ A,\ldots,B\right\} $, and let $\varphi:D\rightarrow\mathbb{Z}^{+}$
be a nondecreasing function. Let $t_{\varphi}$ be an upper bound on the time needed to evaluate $\varphi$. Then,
for every given parameters $\varphi,D$ and $K>1$, function \textsc{ApxSet} computes a $K-$approximation set of
$\varphi$ in $O\left(t_{\varphi}\left(1+\log_{K}\varphi^{\mbox{max}}\right)\log|D|\right)$ time. This $K-$approximation
set has cardinality of $O\left(1+\log_{K}\varphi^{\mbox{max}}\right)$.
\end{prop}
Recalling the principles of Example \ref{exa:}, the function $ApxSet$ enables us to pick a succinct set of representatives,
and build a $K-$approximation function by the arguments in Definition \ref{def:K apx set}.

Sections 5 and 6 in \cite{key-5} provide a set of general computational rules of $K-$approximation sets and functions.
We now survey some of these rules, needed in the proceeding sections.

The validity of the next proposition follows directly from the definition of $K-$approximation functions.
\begin{prop}
\label{prop:5.1}(Based on Proposition 5.1 of \cite{key-5}):For $i=1,2$ let $K_{i}>1$, let $\varphi_{i}:\left\{ A,\ldots,B\right\} \rightarrow\mathbb{Z}^{+}$
and let $\tilde{\varphi}_{i}:\left\{ A,\ldots,B\right\} \rightarrow\mathbb{Z}^{+}$ be a $K_{i}-$approximation
of $\varphi_{i}$. The following properties hold:

\textbf{(1) Summation of approximation:} $\tilde{\varphi}_{1}+\tilde{\varphi}_{2}$ is a $\max\left\{ K_{1},K_{2}\right\} -$approximation
function of $\varphi_{1}+\varphi_{2}$.

\textbf{(2) Approximation of approximation: }If $\varphi_{2}=\tilde{\varphi}_{1}$ then $\tilde{\varphi}_{2}$ is
a $K_{1}K_{2}-$approximation function of $\varphi_{1}$.\end{prop}
\begin{rem}
\label{rem:substruct}Note that there is no rule for subtraction of functions. To illustrate why, let $\varphi$
and $\psi,$ be arbitrary functions, and define $\xi=\varphi-\psi$. If for some $j$ $\varphi\left(j\right)=\psi\left(j\right)$,
we must have that for any $K-$ approximation function $\tilde{\xi}$ of $\xi$, $\tilde{\xi}\left(j\right)=0$.
But the approximation $\tilde{\varphi}\left(j\right)$ may not be equal to the approximation $\tilde{\psi}\left(j\right)$,
so $\tilde{\varphi}-\tilde{\psi}$ is not necessarily a $K-$approximation function of $\xi$.\end{rem}
\begin{prop}
\label{prop:6.2}(Proposition 6.2(3) in \cite{key-5}):For $i=1,2$ let $K_{i}>1$, let $\varphi_{i}:\left\{ A,\ldots,B\right\} \rightarrow\mathbb{Z}^{+}$
be a nondecreasing function. Let $W_{1}$ a $K_{1}-$approximation set of $\varphi_{1}$. Then:

\textbf{Approximation of approximation sets: }If $\varphi_{1}$ is a $K_{2}-$approximation of $\varphi_{2}$, then
$\hat{\varphi}_{1}$ (i.e. the approximation of $\varphi_{1}$ induced by $W_{1}$) is a $K_{1}K_{2}-$approximation
of $\varphi_{2}$.
\end{prop}

\section{Counting $m-$tuples\label{sec:Counting-tuples}}

\subsection{DP formulation}

We first introduce a possible DP formulation. Let 
\[
z_{i}\left(j\right)=\#\left\{ \left(x_{1\ell_{1}},...,x_{i\ell_{i}}\right)|\sum_{k=1}^{i}x_{k\ell_{k}}\geq j\right\} 
\]
 Then the DP formulation is 
\begin{eqnarray}
z_{1}\left(j\right) & =\#\left\{ x_{1k}|1\leq k\leq|X_{1}|,\, x_{1k}\geq j\right\}  & j=0,...,B\nonumber \\
z_{i}\left(j\right) & =\sum_{k=1}^{|X_{i}|}z_{i-1}\left(j-x_{ik}\right) & i=2,...,m;\; j=0,...,B\label{eq:1}\\
z_{i}\left(j\right) & =\prod_{j=1}^{i}\left|X_{j}\right| & i=2,...,m;\ j<0\nonumber 
\end{eqnarray}
The solution is $z_{m}\left(B\right)$. Using this formulation, we can compute $z_{m}\left(B\right)$ in $O\left(B\sum_{i=1}^{m}\left|X_{i}\right|\right)$
time, i.e. in time pseudo polynomial in $B$. For every $1\leq i\leq m$, it is easy to see that $z_{i}\left(j\right)$
is a nonincreasing function. This enables us to use the technique of $K-$approximation sets and functions. 

Note that although in Section 1.4 we presented the method of $K-$approximation sets and functions for nondecreasing
functions, it is easy to see that it is applicable to nonincreasing functions as well.

\subsection{Algorithm}
\begin{prop}
\label{prop:m-tuples struct}Let $\tilde{z}_{i-1}\left(\cdot\right)$ be a nonincreasing $K-$approximation function
of $z_{i-1}\left(\cdot\right)$. Denote 
\[
\bar{z}_{i}\left(j\right)=\sum_{k=1}^{|X_{i}|}\tilde{z}_{i-1}\left(j-x_{ik}\right)
\]
Then $\bar{z}_{i}\left(\cdot\right)$ is a nonincreasing $K-$approximation of $z_{i}\left(\cdot\right)$.\end{prop}
\begin{proof}
The proposition is immediate by the DP formulation, and by Proposition \ref{prop:5.1}(1) (summation of approximation).
$\bar{z}_{i}$ is a sum of nonincreasing functions, and is therefore nonincreasing. 
\end{proof}
Function \textsc{ApxSet} as presented in Section \ref{sub:approximation-sets-and} is used to construct a $K-$approximation
set for nondecreasing functions. It can be modified in an obvious way for nonincreasing functions (see Appendix
\ref{sec:-ApxSet-for}). We consider it as formulated for nonincreasing functions.

We now introduce the algorithm to approximate $z_{i}\left(\cdot\right)$:

\begin{algorithm}[h]
\protect\caption{An FPTAS for calculating the number of solutions for the $m-$tuples problem}

\begin{enumerate}
\item \textbf{Function }\textbf{\textsc{FPTASMtuple}}$\left(\left\{ X_{i},\ldots,X_{m}\right\} ,B,\epsilon\right)$
\item $K$$\leftarrow$\textcolor{black}{$\sqrt[m]{1+\epsilon}$,} $W_{1}\leftarrow$\textsc{ApxSet}\textbf{$\left(z_{1},\left\{ 0,...,B\right\} ,K\right)$}
\item let $\hat{\bar{z}}_{1}$ be the approximation of $z_{1}$ indued by $W_{1}$
\item for i:=2 to $m$ 

\begin{enumerate}
\item let $\bar{z}_{i}\left(j\right)=\sum_{k=1}^{|X_{i}|}\hat{\bar{z}}_{i-1}\left(j-x_{ik}\right)$
\item $W_{i}\leftarrow$\textsc{ApxSet}$\left(\bar{z}_{i},\left\{ 0,...,B\right\} ,K\right)$
\item let $\hat{\bar{z}}_{i}$ be the approximation of $\bar{z}_{i}$ induced by $W_{i}$
\end{enumerate}
\item end for
\item \textbf{return} $\hat{\bar{z}}_{m}\left(B\right)$\end{enumerate}
\end{algorithm}

Before proving this algorithm's performance, we give an example for the operation of the algorithm:
\begin{example}
Let 
\begin{eqnarray*}
X_{1} & = & \left\{ 1,3,7\right\} \\
X_{2} & = & \left\{ 2,5\right\} \\
X_{3} & = & \left\{ 3,9\right\} 
\end{eqnarray*}
and let $B=17$.

By (\ref{eq:1}) the exact functions $z_{1}z_{2},z_{3}$ are:

\begin{tabular}{|c|c|c|c|c|c|c|c|c|c|c|c|c|c|c|c|c|c|c|}
\hline 
 & 0 & 1 & 2 & 3 & 4 & 5 & 6 & 7 & 8 & 9 & 10 & 11 & 12 & 13 & 14 & 15 & 16 & 17\tabularnewline
\hline 
\hline 
$z_{1}$ & 3 & 3 & 2 & 2 & 1 & 1 & 1 & 1 & 0 & 0 & 0 & 0 & 0 & 0 & 0 & 0 & 0 & 0\tabularnewline
\hline 
$z_{2}$ & 6 & 6 & 6 & 6 & 5 & 5 & 4 & 3 & 3 & 2 & 1 & 1 & 1 & 0 & 0 & 0 & 0 & 0\tabularnewline
\hline 
$z_{3}$ & 12 & 12 & 12 & 12 & 12 & 12 & 12 & 11 & 11 & 10 & 9 & 9 & 8 & 6 & 6 & 5 & 3 & 3\tabularnewline
\hline 
\end{tabular}\\

Now, suppose we execute Algorithm 2 with $\epsilon=7$, i.e. $K=\sqrt[3]{1+7}=2$ and the output is guaranteed to
provide a $8-$approximation.

The $2-$approximation set of the function $z_{1}$ is $W_{1}=\left\{ 0,4,8,17\right\} $, and the $2-$approximation
function induced by $W_{1}$ is :

\begin{tabular}{|c|c|c|c|c|c|c|c|c|c|c|c|c|c|c|c|c|c|c|}
\hline 
 & 0 & 1 & 2 & 3 & 4 & 5 & 6 & 7 & 8 & 9 & 10 & 11 & 12 & 13 & 14 & 15 & 16 & 17\tabularnewline
\hline 
\hline 
$\hat{\bar{z}}_{1}$ & 3 & 3 & 3 & 3 & 1 & 1 & 1 & 1 & 0 & 0 & 0 & 0 & 0 & 0 & 0 & 0 & 0 & 0\tabularnewline
\hline 
\end{tabular}\\

We turn now to the for-loop. The first iteration is for $i=2$:

\begin{tabular}{|c|c|c|c|c|c|c|c|c|c|c|c|c|c|c|c|c|c|c|}
\hline 
 & 0 & 1 & 2 & 3 & 4 & 5 & 6 & 7 & 8 & 9 & 10 & 11 & 12 & 13 & 14 & 15 & 16 & 17\tabularnewline
\hline 
\hline 
$\bar{z}_{2}$ & 6 & 6 & 6 & 6 & 6 & 6 & 4 & 4 & 4 & 2 & 1 & 1 & 1 & 0 & 0 & 0 & 0 & 0\tabularnewline
\hline 
\end{tabular}\\

Note that the algorithm does not compute $\bar{z}_{2}$ over its entire domain, but only over the values needed
for constructing $W_{2}$. Note also that $\bar{z}_{2}$ is $2-$approximation of $z_{2}$ (In fact in this example
it is $1\frac{1}{3}-$approximation). Now, $W_{2}=\left\{ 0,9,13,17\right\} $ is a $2-$approximation set of $\bar{z}_{2}$
. The function induced by $W_{2}$ is:

\begin{tabular}{|c|c|c|c|c|c|c|c|c|c|c|c|c|c|c|c|c|c|c|}
\hline 
 & 0 & 1 & 2 & 3 & 4 & 5 & 6 & 7 & 8 & 9 & 10 & 11 & 12 & 13 & 14 & 15 & 16 & 17\tabularnewline
\hline 
\hline 
$\hat{\bar{z}}_{2}$ & 6 & 6 & 6 & 6 & 6 & 6 & 6 & 6 & 6 & 2 & 2 & 2 & 2 & 0 & 0 & 0 & 0 & 0\tabularnewline
\hline 
\end{tabular}\\

This is a $4-$approximation function of $z_{2}$ (In fact $2-$approximation). 

We turn now to the next iteration for $i=3$:

\begin{tabular}{|c|c|c|c|c|c|c|c|c|c|c|c|c|c|c|c|c|c|c|}
\hline 
 & 0 & 1 & 2 & 3 & 4 & 5 & 6 & 7 & 8 & 9 & 10 & 11 & 12 & 13 & 14 & 15 & 16 & 17\tabularnewline
\hline 
\hline 
$\bar{z}_{3}$ & 12 & 12 & 12 & 12 & 12 & 12 & 12 & 12 & 12 & 12 & 12 & 12 & 8 & 8 & 8 & 8 & 6 & 6\tabularnewline
\hline 
\end{tabular}\\

This is $4-$approximation of $z_{3}$. This function is used to calculate a $2-$approximation set of $\bar{z}_{3}$:
$W_{3}=\left\{ 0,17\right\} $, so the function induced by $W_{3}$ is:

\begin{tabular}{|c|c|c|c|c|c|c|c|c|c|c|c|c|c|c|c|c|c|c|}
\hline 
 & 0 & 1 & 2 & 3 & 4 & 5 & 6 & 7 & 8 & 9 & 10 & 11 & 12 & 13 & 14 & 15 & 16 & 17\tabularnewline
\hline 
\hline 
$\hat{\bar{z}}_{3}$ & 12 & 12 & 12 & 12 & 12 & 12 & 12 & 12 & 12 & 12 & 12 & 12 & 12 & 12 & 12 & 12 & 12 & 12\tabularnewline
\hline 
\end{tabular}\\

The function $\hat{\bar{z}}_{3}$ is a $8-$approximation function of $z_{3}$ (In fact $4-$approximation). 

We get the approximated value $\hat{\bar{z}}_{3}\left(17\right)=12$, while the exact solution is $z_{3}\left(17\right)=3$.\end{example}
\begin{prop}
Let $0<\epsilon<1$, and let $X_{1},\ldots,X_{m},B\in\mathbb{N}$ be an instance of the $m-$tuples problem. Then
$\hat{\bar{z}}_{m}\left(B\right)$ calculated by Algorithm 2 is a $\left(1+\epsilon\right)-$approximation function
of $z_{m}\left(B\right)$ . The algorithm is deterministic and runs in time 
\[
O\left(\frac{m^{2}}{\epsilon}\left(\sum_{i=1}^{m}|X_{i}|\right)\log\left(\Pi_{i=1}^{m}|X_{i}|\right)\log\left(\frac{m\log\Pi_{i=1}^{m}|X_{i}|}{\epsilon}\right)\log B\right)
\]
.\end{prop}
\begin{proof}
\textbf{Correctness: }By its definition in (\ref{eq:1}), $z_{1}$ is nonincreasing, so the call to \textsc{ApxSet}
in step 2 is well defined. Therefore $\hat{\bar{z}}_{1}$ is a nonincreasing function as a function induced by a
$K-$approximation set of a nonincreasing function.

It is easy to see by induction that $\hat{\bar{z}}_{i}$ is a nonincreasing function: By Proposition \ref{prop:m-tuples struct}
$\bar{z}_{i}$ is nonincreasing, so the call to \textsc{ApxSet} is well defined. $\hat{\bar{z}}_{i}$ is a function
induced by a $K-$approximation set for a nonincreasing function, and thus is a nonincreasing function. 

We next consider the approximation ratio. We first show by induction that: 

(1) $\hat{\bar{z}}_{i}$ is a nonincreasing $K^{i}-$approximation function of $z_{i}$.

(2) $W_{i}$ is a $K-$approximation set of $\bar{z}_{i}$. 

Base case: by Proposition \ref{prop:4.6}, with parameters set to $\varphi=z_{1},\, D=\left\{ 0,\ldots,B\right\} ,$
and $K=K$, $W_{1}$ is a $K-$approximation set of $\bar{z}_{1}\left(\equiv z_{1}\right)$. By Proposition \ref{prop:4.5}
$\hat{\bar{z}}_{1}$ is a $K-$approximation function of $z_{i}$, so (1) and (2) hold for $i=1$.

Assume (1)+(2) hold for $i-1$. By the induction hypothesis (1), $\hat{\bar{z}}_{i-1}$ is a nonincreasing $K^{i-1}-$approximation
function of $z_{i-1}$. Then by Proposition \ref{prop:m-tuples struct} $\bar{z}_{i}$ is a nonincreasing $K^{i-1}-$approximation
function of $z_{i}$. By Proposition \ref{prop:4.6}, with parameters set to $\varphi=\bar{z}{}_{i},\, D=\left\{ 0,\ldots,B\right\} ,$
and $K=K$, $W_{i}$ is a $K-$approximation set of $\bar{z}_{i}$. By Proposition \ref{prop:5.1}(2) (approximation
of approximation) with parameters set to $\varphi_{1}=z_{i},\,\varphi_{2}=\bar{z}_{i},\,\tilde{\varphi}_{2}=\hat{\bar{z}}_{i},\, K_{1}=K^{i-1}$,
and $K_{2}=K$, we get that $\hat{\bar{z}}_{i}$ is a nonincreasing (Definition \ref{def:K apx set}) $K^{i}-$approximation
of $z_{i}.$ This completes the proof by induction.

Recall that $K=\sqrt[m]{1+\epsilon}$. We deduce from (1) above with $i=m$, that for every $0\leq j\leq B$ we
have $z_{m}\left(j\right)\leq\hat{\bar{z}}_{m}\left(j\right)\leq\left(\left(\sqrt[m]{1+\epsilon}\right)^{m}\right)z_{m}\left(j\right).$
When $j=B$, we therefore have $z_{m}\left(B\right)\leq\hat{\bar{z}}_{m}\left(B\right)\leq\left(1+\epsilon\right)z_{m}\left(B\right)$.
This proves the approximation ratio.

\textbf{Running time:} The running time of the algorithm is dominated by the for-loop that has $m$ iterations.
Every iteration is dominated by the call to \textsc{ApxSet}. By Proposition \ref{prop:4.6} the running time is
$O\left(t_{\bar{z}_{i}}\log_{K}\Pi_{i=1}^{m}|X_{i}|\log B\right)$ (Note that if every $m-$tuple is a feasible
solution, then there are $\prod_{i=1}^{m}\left|X_{i}\right|$ solutions). $\hat{\bar{z}}_{i}$ can be stored efficiently
(as a function induced by a $K-$approximation set). Thus, by Proposition \ref{prop:4.5} $t_{\hat{\bar{z}}_{i}}=O\left(\log\log_{K}\Pi_{i=1}^{m}|X_{i}|\right)$.
Thus by the definition of $\bar{z}\left(\cdot\right)$, $t_{\bar{z}_{i}}=O\left(|X_{i}|\log\log_{K}\Pi_{i=1}^{m}|X_{i}|\right)$.
Then the running time of the $i$'th iteration is $O\left(|X_{i}|\log\log_{K}\Pi_{i=1}^{m}|X_{i}|\log_{K}\Pi_{i=1}^{m}|X_{i}|\log B\right)$.
Using the fact that $O\left(\log_{K}\prod_{i=1}^{m}\left|X_{i}\right|\right)=O\left(\frac{m\log\prod_{i=1}^{m}\left|X_{i}\right|}{\log\left(1+\epsilon\right)}\right)=O\left(\frac{m\log\prod_{i=1}^{m}\left|X_{i}\right|}{\epsilon}\right)$,
which holds true by the inequality $\epsilon\leq\log\left(1+\epsilon\right)$, which holds for every $0\leq\epsilon\leq1$.
We thereby conclude that the running time is 
\[
O\left(\frac{m^{2}}{\epsilon}\left(\sum_{i=1}^{m}|X_{i}|\right)\log\left(\Pi_{i=1}^{m}|X_{i}|\right)\log\left(\frac{m\log\Pi_{i=1}^{m}|X_{i}|}{\epsilon}\right)\log B\right)
\]
.
\end{proof}
As far as we know, this is the first FPTAS for this problem.

\section{contingency tables with 2 rows\label{sec:contingency-tables-with}}

The problem of approximately counting contingency tables with 2 rows was considered by Dyer and Greenhill , who
developed a random algorithm to solve it \cite{key-2}. Dyer developed a strongly polynomial random algorithm for
the general case of $m$ rows \cite{key-11}. Gopalon et al developed the first FPTAS for the general problem \cite{key-3}.
Discussion about the running times of these algorithms is given in the end of this section. The following algorithm
is faster than any of the former three, and is relatively simple.

\subsection{First DP formulation}

The general problem for $m$ rows is introduced in Dyer and Greenhill \cite{key-2}. We want to develop an FPTAS
for calculating $\left|\Sigma_{s,r}\right|$ when $m=2,$ i.e. $r=\left(r_{1},r_{2}\right).$ Dyer and Greenhill
offer the following dynamic programming formula with it $\left|\Sigma_{s,r}\right|$ can be calculated:

Let $R=\min\left\{ r_{1},r_{2}\right\} $. The input size is therefore $O\left(\log N+\log R+\sum_{i=1}^{n}\log s_{i}\right).$
For $1\leq j\leq R,\,1\leq i\leq n$ let
\[
\mathcal{G}_{i}\left(j\right)=\left\{ \left(x_{1},...,x_{i}\right)\in\mathbb{Z}^{+^{i}}:\sum_{k=1}^{i}x_{k}=j\mbox{ and }0\leq x_{k}\leq s_{k}\mbox{ for }1\leq k\leq i\right\} 
\]
Let $A_{i}\left(j\right)=\left|\mathcal{G}_{i}\left(j\right)\right|.$ Then the objective function is $\left|\Sigma_{rs}\right|=A_{n}\left(R\right).$
The boundary conditions are $A_{i}\left(0\right)=1$ for $0\leq i\leq n,$ and $A_{0}\left(j\right)=0$ for $1\leq j\leq R.$
Dyer and Greenhill presented the following recurrence: 
\begin{equation}
A_{i}\left(j\right)=\begin{cases}
A_{i-1}\left(j\right)+A_{i}\left(j-1\right) & j-1<s_{i},\\
A_{i-1}\left(j\right)+A_{i}\left(j-1\right)-A_{i-1}\left(j-1-s_{i}\right) & j-1\geq s_{i}.
\end{cases}\label{eq:DPformula1}
\end{equation}

$A_{n}\left(R\right)$ can be computed in $O\left(nR\right)$ time.

An explanation for this formula is as follow: suppose $j-1<s_{i}$, and we want to assign $j$ (identical) items
into cells $1,...,i$. There are 2 cases:

Case 1: The $i$'th cell is empty. Then there are $A_{i-1}\left(j\right)$ combinations for the assignment in cells
$1,\ldots,i-1$.

Case 2: There is at least one item in the $i$'th cell. We put item $j$ in cell $i$, and then there are $A_{i}\left(j-1\right)$
combinations to assign items $1,...,j-1$ into cells $1,\ldots,i$ (Note: $j-1<s_{i}$, so there is no restriction
on the number of items to put in the $i$'th cell).

Suppose now $j-1\geq s_{i}$. Again, there are 2 cases:

Case 1: The $i$'th cell is empty. There are $A_{i-1}\left(j\right)$ combinations as before.

Case 2: There is at least one item in the $i$'th cell. We put item $j$ in cell $i$, and the number of combinations
is as before ($A_{i}\left(j-1\right)$), but we have to preclude the case where the $i$'th cell contains $s_{i}$
items from items $1,\ldots,j-1$ as well as item j, i.e. to subtract $A_{i-1}\left(j-1-s_{i}\right)$.

Now, we prove a proposition about the structure of function $A_{i}\left(\cdot\right),\,1\leq i\leq n$:
\begin{prop}
\label{prop:contingency struct}For every $i=1,\ldots,n$ let $B_{i}=\sum_{l=1}^{i}s_{i}$. The following two properties
hold:

(1) $A_{i}\left(\cdot\right)$ is symmetric around $\frac{B_{i}}{2}$ in the range $\left\{ 0,\ldots,B_{i}\right\} $.
i.e. $A_{i}\left(j\right)=A_{i}\left(B_{i}-j\right)$ for $j=0,\ldots,\lfloor\frac{B_{i}}{2}\rfloor$

(2) $A_{i}\left(\cdot\right)$ is unimodal in the following way: $A_{i}\left(j\right)$ is nondecreasing for $j=0,\ldots,\lfloor\frac{B_{i}}{2}\rfloor$,
is nonincreasing for $j=\lceil\frac{B_{i}}{2}\rceil,\ldots,B_{i}$ and $A_{i}\left(j\right)=0$ for $j>B_{i}$.\end{prop}
\begin{proof}
We start with Property (1). Let $0\leq j\leq\lfloor\frac{B_{i}}{2}\rfloor$. For any assignment of $j$ items in
cells $1,\ldots,i$ in row 1, switching between row 1 and 2 gives an assignment of $B_{i}-j$ items in cells $1,\ldots,i$
in row 1. This provide us a one to one correspondence between assignments of $j$ items in cells $1,\ldots,i$ in
row 1, and assignments of $B_{i}-j$ items in cells $1,\ldots,i$ in row 1. This completes the proof of the first
property.

We now turn to Property (2). For $j>B_{i}$, the cells cannot contain the items, so there are no valid assignments,
therefore $A_{i}\left(j\right)=0$. To complete the proof, it is enough to prove that $A_{i}\left(j\right)$ is
nondecreasing for $j=1,\ldots,\lfloor\frac{B_{i}}{2}\rfloor$. The other part of the statement is immediate by the
symmetry of the function.

We now prove by induction on $i$ that $A_{i}\left(\cdot\right)$ is nondecreasing over $\left\{ 1,...,\lfloor\frac{B_{i}}{2}\rfloor\right\} $.
Considering the base case of $i=1$,we note that $A_{1}\left(\cdot\right)\equiv1$ on $\left\{ 1,\ldots,B_{1}\right\} $,
so the function is nondecreasing in the relevant range. This proves the base case.

The induction hypothesis for $i-1$ is that $A_{i-1}\left(\cdot\right)$ is nondecreasing over $\left\{ 1,...,\lfloor\frac{B_{i-1}}{2}\rfloor\right\} $.
Let $j$ be such that $2\leq j\leq\lfloor\frac{B_{i}}{2}\rfloor$. We need to show that $A_{i}\left(j\right)\geq A_{i}\left(j-1\right)$.

Case 1: $j-1<s_{i}$. By (\ref{eq:DPformula1}), $A_{i}\left(j\right)=A_{i}\left(j-1\right)+A_{i-1}\left(j\right)$.
The proof follows due to the nonnegativity of $A_{i-1}\left(j\right)$.

Case 2: $j-1\geq s_{i}$. By (\ref{eq:DPformula1}), $A_{i}\left(j\right)-A_{i}\left(j-1\right)=A_{i-1}\left(j\right)-A_{i-1}\left(j-1-s_{i}\right)$.
It therefore remains to show that $A_{i-1}\left(j\right)\geq A_{i-1}\left(j-1-s_{i}\right)$. Now:

\begin{eqnarray*}
j\leq\lfloor\frac{B_{i}}{2}\rfloor & \Rightarrow\\
j\leq\frac{B_{i}}{2}+\frac{1}{2} & \Rightarrow\\
2j\leq B_{i-1}+s_{i}+1 & \Rightarrow\\
j-\frac{B_{i-1}}{2}\leq\frac{B_{i-1}}{2}-\left(j-s_{i}-1\right) & \Rightarrow\\
\left|\frac{B_{i-1}}{2}-j\right|\leq\left|\frac{B_{i-1}}{2}-\left(j-s_{i}-1\right)\right|
\end{eqnarray*}
By the symmetry of $A_{i-1}$ around $\frac{B_{i-1}}{2}$ and by the unimodality of $A_{i-1}$, we get that $A_{i-1}\left(j\right)\geq A_{i-1}\left(j-1-s_{i}\right)$
as required. 
\end{proof}
There are two issues that prevent us from using the method of $K-$approximation sets and functions with DP formulation
(\ref{eq:DPformula1}). The first is that the formulation involves subtraction (see Remark \ref{rem:substruct}).
The second is that in the method, the evaluation of $A_{i}\left(\cdot\right)$ needs to rely only on the functions
evaluated before, i.e. on $A_{j}\left(\cdot\right)$ for $j<i$. But this is not the case in this formulation.

Therefore we need to turn to another DP formulation.

\subsection{Second DP formulation}

We introduce now another formulation, that better suits our purpose. The boundary conditions and the objective function
are the same as before:
\begin{equation}
A_{i}\left(j\right)=\sum_{k=0}^{\min\left(j,s_{i}\right)}A_{i-1}\left(j-k\right)\label{eq:DPformula2}
\end{equation}

$A_{n}\left(R\right)$ can be computed in $O\left(nR\max_{1\leq i\leq n}s_{i}\right)$ time.

\emph{Explanation:} For counting the number of combinations to assign $j$ items into $i$ cells, we sum over the
number of items in the $i$'th cell. If $s_{i}\geq j$, the $i'$th cell contains $0,\ldots,j$ items. If $s_{i}<j$,
the $i$'th cell contains $0,\ldots,s_{i}$ items.
\begin{rem}
One can deduce the second formulation from the first formulation by induction on $j$. The base case is immediate.
We show the case when $s_{i}\leq j-1$ (The other case is simple and is therefore omitted) . The induction hypothesis
is $A_{i}\left(j-1\right)=\sum_{k=0}^{s_{i}}A_{i-1}\left(j-1-k\right)$. Now:
\begin{eqnarray*}
A_{i}\left(j\right) & = & A_{i-1}\left(j\right)+A_{i}\left(j-1\right)-A_{i-1}\left(j-1-s_{i}\right)\\
 & = & \sum_{k=0}^{s_{i}}A_{i-1}\left(j-1-k\right)+A_{i-1}\left(j\right)-A_{i-1}\left(j-1-s_{i}\right)\\
 & = & \sum_{k=0}^{s_{i}-1}A_{i-1}\left(j-1-k\right)+A_{i-1}\left(j\right)\\
 & = & \sum_{k=1}^{s_{i}}A_{i-1}\left(j-k\right)+A_{i-1}\left(j\right)\\
 & = & \sum_{k=0}^{s_{i}}A_{i-1}\left(j-k\right)
\end{eqnarray*}

\end{rem}
This formulation does not fit the method of $K-$approximation sets and functions, since $\min\left(s_{i},j\right)$
could be of order $O\left(R+s_{i}\right)$, i.e. exponential in the input size. Thus we need to turn to a third
DP formulation.

\subsection{Third DP formulation}

The difficulty which arises in formulation (\ref{eq:DPformula2}) also arises in \cite{key-4} in the context of
counting integer knapsack solutions: Given $n$ elements with nonnegative integer weights $w_{1},\ldots,w_{n}$,
an integer capacity $C$, and positive integer ranges $u_{1},\ldots,u_{n}$, find the cardinality of the set of
solutions $\left\{ x\in\mathbb{Z}^{+}|\sum_{i=1}^{n}w_{i}x_{i}\leq C,\,0\leq x_{i}\leq u_{i}\right\} $. The following
arguments are very similar to the ones in \cite{key-4}. 

We evaluate $A_{i}\left(\cdot\right)$ only over $\left\{ 0,\ldots,\lfloor\frac{B_{i}}{2}\rfloor\right\} $, where
it is nondecreasing. The approximation on the entire domain is clear by Proposition \ref{prop:contingency struct}
(see the details in Algorithm 3). 

We next introduce the function \textsc{CompressContingency}, which is a version of the function Compress in \cite{key-4}.

\begin{algorithm}[h]
\protect\caption{Returns a step-wise $K-$approximation of $\varphi$}

\begin{enumerate}
\item \textbf{Function }\textbf{\textsc{CompressContingency}}\textbf{$\left(\varphi,K,B_{i}\right)$}
\item obtain a $K-$approximation set $W$ of $\varphi$ on $\left\{ 0,...,\frac{\lfloor B_{i}\rfloor}{2}\right\} $
\item Let $\hat{\varphi}$ be the approximation of $\varphi$ induced by $W$
\item Let $\hat{\varphi}\left(j\right)=\hat{\varphi}\left(B_{i}-j\right)$ for $j=\lceil\frac{B_{i}}{2}\rceil,...,B_{i}$,
and $\hat{\varphi}\left(j\right)\equiv0$ for $j>B_{i}$
\item \textbf{return $\hat{\varphi}$}\end{enumerate}
\end{algorithm}

The next proposition is similar to Proposition 2.2 in \cite{key-4}. It is deduced by Propositions \ref{prop:4.5}-\ref{prop:5.1}
and \ref{prop:contingency struct} above.
\begin{prop}
Let $K_{1},K_{2}\geq1$ be real numbers, $M>1$ be an integer, and let $\varphi:\left[0,...,B\right]\rightarrow\left[0,...,M\right]$
be a function with structure as in Proposition 3.1. Let $\bar{\varphi}$ be a $K_{2}-$approximation function of
$\varphi$. Then function $CompressContingency\left(\bar{\varphi},K_{1},B_{i}\right)$ returns in $O\left(\left(1+t_{\bar{\varphi}}\right)\left(\log_{K_{1}}M\log B\right)\right)$
time a piecewise step function $\hat{\varphi}$ with structure as in Proposition \ref{prop:contingency struct},
with $O\left(\log_{K_{1}}M\right)$ pieces, which $K_{1}K_{2}-$approximates $\varphi$. The query time of $\hat{\varphi}$
is $O\left(\log\log_{K_{1}}M\right)$ if it is sorted in a sorted array $\left\{ \left(x,\hat{\varphi}\right)|x\in W\right\} $. 
\end{prop}
We now give a third dynamic programming formulation, which is pseudo-polynomial in the size of $R$ only. Denote
by $m_{i}\left(j\right)=\min\left(j,s_{i}\right),\, w_{i}=1$ for every $i=1,\ldots,n.$ Then apart from the fifth
equation, our DP formulation is identical to formulation (2) in \cite{key-4}.

An explanation for the following formulation is: In (\ref{eq:DPformula2}) the evaluation of $A_{i}$$\left(j\right)$
is done at once by summing over all the possible values of the number of items in the $i$'th cell. In the following
formulation we break this evaluation into $\lfloor\log m_{i}\left(j\right)\rfloor+1$ separate simple evaluations:
In the $\ell$'th evaluation we look at the $\ell$'th digit of the binary representation of $m_{i}\left(j\right)$
and consider it to be 0 or 1, i.e. consider to put or not to put $2^{\ell-1}$ items in the $i$'th cell. For every
of these options we can calculate the number of contingency tables when using cells $1,\ldots,i$ only, and in the
$i$'th cell there are no more than $s_{i}\mod2^{\ell}$ items. After considering this two options for all the digits
in the binary representation of $m_{i}\left(j\right)$, we get $A_{i}\left(j\right)$.

For doing this Halman introduces the idea of \emph{binding constraints} \cite{key-4}. For $\ell\geq1$ let $z_{i,\ell,r}\left(j\right)$
be the number of solutions for contingency tables with with 2 rows, where $j$ items are in the first row, that
use cells $\left\{ 1,\ldots i\right\} $, put no more than $s_{i}\mod2^{\ell}$ items in the $i$'th cell, and no
more than $s_{k}$ items in the $k$'th cell, for $k=1,\ldots,i-1$. In this way, the future assignments can affect
on the current assignment: The number of items we can assign in the $\ell$'th step into the $i$'th cell affected
by the number of items will assign in the next steps. i.e. the number of assignments in the $\ell$'th step depends
on the question if the number of items will assign in the next steps will leave enough capacity for assign as many
items as we want, or it is cause a restriction on the number of items in the current step. We need to consider the
both options, and this is done by the third index of $z_{i,\ell,r}\left(j\right)$: If $r=0$ then the constraint
of having no more than $s_{i}$ items in the $i$'th cell is assumed to be non binding (i.e. we assume there is
enough capacity for $2^{\ell}-1$ more items in the $i$'th cell). If, on the other hand, $r=1$ then this constraint
may be binding. E.g. if $s_{i}=5$ and $\ell=2$, and there are already 4 items in cell $i$, we are in the case
of $r=1$, since there is not remaining capacity for $2^{2}-1=3$ additional items. If the $i$'th cell is empty,
we are in the case of $r=0$, since there is enough capacity for 3 more items.

Let us introduce some definitions before giving the formal recursion: Let $\log^{+}\left(x\right)$ equal $\log\left(x\right)$
for $x\geq1$ and 0 otherwise. Let $\mbox{msb}\left(x,i\right):=\lfloor\log\left(x\mbox{ mod }2^{i}\right)\rfloor+1$.
$\mbox{msb}\left(x,i\right)$ is therefore the most significant 1-digit of $\left(x\mbox{ mod }2^{i}\right)$ if
$\left(x\mbox{ mod }2^{i}\right)>0$, and is $-\infty$ otherwise. E.g., $\mbox{msb}\left(5,2\right)=1$ and $\mbox{msb}\left(4,1\right)=-\infty$.
\begin{eqnarray}
z_{i,\ell,0}\left(j\right) & =z_{i,\ell-1,0}\left(j\right)+z_{i,\ell-1,0}\left(j-2^{\ell-1}\right)\label{eq:DP3.1}\\
z_{i,\ell,1}\left(j\right) & =z_{i,\ell-1,0}\left(j\right)+z_{i,\mbox{msb}\left(s_{i},\ell-1\right),0}\left(j-2^{\ell-1}\right)\label{eq:DP3.2}\\
z_{i,1,r}\left(j\right) & =z_{i-1,\lfloor\log^{+}m_{i-1}\left(j\right)\rfloor+1,1}\left(j\right)+\nonumber \\
 & +z_{i-1,\lfloor\log^{+}m_{i-1}\left(j-1\right)\rfloor+1,1}\left(j-1\right)\label{eq:DP3.3}\\
z_{i,-\infty,1}\left(j\right) & =z_{i-1,\lfloor\log^{+}m_{i-1}\left(j\right)\rfloor+1,1}\left(j\right)\label{eq:DP3.4}\\
z_{1,\ell,r}\left(j\right) & =1\label{eq:DP3.5}\\
z_{i,\ell,r}\left(j\right) & =0 & j<0\label{eq:DP3.6}
\end{eqnarray}
where $r=0,1$ , $i=2,\ldots,n$ , $\ell=2,\ldots,\lfloor\log^{+}m_{i}\left(j\right)\rfloor+1$ , and $j=0,\ldots,R$
unless otherwise specified. The objective function is $z_{n,\lfloor\log s_{n}\rfloor+1,1}\left(R\right).$ Denote
$S=\max_{1\leq i\leq n}s_{i}$, so the complexity of this pseudo-polynomial algorithm is $O\left(nR\log S\right)$.

We now turn to a more detailed explanation of formulations (\ref{eq:DP3.1})-(\ref{eq:DP3.6}). In the case of equation
(\ref{eq:DP3.1}) we assume that there is enough capacity for putting $2^{\ell}-1$ more items in the $i$'th cell,
and therefore, in both cases of the values of the $\ell$'th bit, there is still enough capacity in the $i$'th
cell for as many items as we want.

In equation (\ref{eq:DP3.2}) we assume the constraint of having no more than $s_{i}$ items in the $i$'th cell
may be binding. So when putting $2^{\ell-1}$ items in the $i$'th cell, we have to take the constraint into account.
If we do not put $2^{\ell-1}$ items, clearly the constraint will not be binding anymore.

The remaining four equations deal with boundary conditions: Equation (\ref{eq:DP3.3}) deals with the case of $\ell=1$,
i.e. the possibility of having an odd number of items. Equation (\ref{eq:DP3.4}) can be called by (\ref{eq:DP3.2})
when there are exactly $s_{i}$ items in the $i$'th cell, or by (\ref{eq:DP3.3}) when $m_{i}\left(j\right)=0$,
i.e. there is not enough capacity to put even a single item. Equation (\ref{eq:DP3.5}) deals with the base case
of one item, and the last equation deals with the boundary condition that there is not enough capacity in the $i$'th
cell.

\subsection{Algorithm}

Now, we can proceed exactly as in section 3.2 of \cite{key-4}, and use the algorithm $\mbox{CountIntegerKnapsack}\left(w,C,u,\epsilon\right)$
with the following notations:
\begin{itemize}
\item $w_{i}=1$ for every $i$
\item $C=R$
\item $u_{i}=s_{i}$
\item Denote $S=\max_{1\leq i\leq n}s_{i}$. Then $U=S$.
\item Include $B_{i}$ in the input of the algorithm
\item Use CompressContingency instead of Compress
\end{itemize}
With this notations, the algorithm analysis is also valid (using Proposition \ref{prop:contingency struct} above
instead of Proposition 2.2 in \cite{key-3}). The running time is $O\left(\frac{\left(n\log S\right)^{3}}{\epsilon}\log\frac{n\log S}{\epsilon}\log R\right)$
.

\subsection{Comparison with other known algorithms}

In Section 3 of \cite{key-2} Dyer and Greenhill introduce a random algorithm based on mixing Markov chains to approximate
the number of contingency tables with two rows. We now outline the analysis of its running time. By page 269 of
\cite{key-2}, let $d=\sum_{k=3}^{n}\lceil\log s_{k}\rceil$, $M=\lceil150e^{2}\frac{d^{2}}{\epsilon^{2}}\log\frac{3d}{\delta}\rceil=O\left(\frac{d^{2}}{\epsilon^{2}}\log\frac{d}{\delta}\right)$
, where the approximation ratio is guaranteed in probability of $1-\delta$, and ,$T=\tau\left(\frac{\epsilon}{15de^{2}}\right)$,
where $\tau\left(\epsilon\right)$ is the mixing time of the Markov chain. By page 270 in \cite{key-2}, the running
time of the algorithm is $O\left(dMT\right)$.

Now, according to Theorem 4.1 in \cite{key-2}, $\tau\left(\epsilon\right)=O\left(n^{2}\log\frac{N}{\epsilon}\right)$.
$M=O\left(\frac{d^{2}}{\epsilon^{2}}\log d\right)$ , and by page 269 $d=O\left(n\log N\right)$. So the running
time is $O\left(n\log\left(N\right)\frac{d^{2}}{\epsilon^{2}}\log\left(d\right)n^{2}\log\frac{dN}{\epsilon}\right)=O\left(\frac{n^{5}}{\epsilon^{2}}\log^{3}\left(N\right)\log\left(n\log N\right)\log\left(\frac{Nn\log N}{\epsilon}\right)\right)$.\textcolor{red}{{}
}This algorithm is slower than our's by a factor of $\frac{n^{2}}{\epsilon}$, up to log terms.

Dyer introduces a randomized algorithm for the general case of contingency tables with $m$ rows, which is strongly
polynomial \cite{key-11}. The running time of this algorithm is $O\left(n^{4m+1}+\frac{n^{3m}}{\epsilon^{2}}\right)$.
For $m=2$ the running time is $O\left(n^{9}+\frac{n^{6}}{\epsilon^{2}}\right)$. Our algorithm is faster by a factor
of at least $O\left(\frac{n^{3}}{\epsilon}\right)$ up to log terms, and is deterministic, but not strongly polynomial.

Gopalan et al give an FPTAS for contingency tables with $m$ rows \cite{key-3}. This algorithm is not strongly
polynomial, and it is slower than Dyers in both $n$ and $\frac{1}{\epsilon}$. In addition, this algorithm is ``fairly
intricate and involve a combination of Dyer's FPRAS for counting contingency tables and our algorithms for counting
general integer knapsack solutions and counting knapsack solutions under small space sources'' (This quotation
is taken from Appendix B in \cite{key-3}). Our algorithm is relatively simple, and runs faster by a factor of at
least $O\left(\frac{n^{12}}{\epsilon}\right)$, but slower by a factor of at most $\log^{2}R$.

\section{Strongly polynomial algorithm for counting 0/1 knapsack\label{sec:0/1-knapsack--strongly}}

A strongly polynomial algorithm for solving the problem of counting 0/1 knapsack solutions introduced by Štefankovi\v{c}
et al \cite{key-6}. We introduce a simple alternative algorithm, with the same running time (up to log term).

\subsection{$K-$approximation set of increasing points}

Let $w_{1},..,w_{n},C$ be an instance of 0/1 knapsack. Let $s_{i}\left(j\right)$ be as defined in (\ref{eq:knapsack set}). 

We want to approximate $s_{n}\left(C\right)$ by a strongly polynomial algorithm. By using the method of $K-$approximation
sets and functions, we can construct an $FPTAS$ for calculating $s_{n}\left(C\right).$ This algorithm runs in
time $O\left(\frac{n^{3}}{\epsilon}\log\frac{n}{\epsilon}\log C\right),$ see Appendix \ref{sec:Not-strongly-polynomial}.
Note that the running time depends on $\log C,$ because the cardinality of the domain of the function is $C.$
In order to avoid this dependence, we first look at $s_{i}$ restricted to the points where it is strictly increasing.
The domain of the restricted function (we call it $s_{i}^{\mbox{inc}}$), affected by the number of the strictly
increasing points of $s_{i}\left(\cdot\right),$ is of size no more than $2^{n}.$ For $s_{i}^{\mbox{inc}}$, we
can construct a $K-$approximation set in strongly polynomial running time. Then we use the set we found to approximate
the original function $s_{i}\left(\cdot\right)$.

We start with some definitions to formalize the definition of $s_{i}^{\mbox{inc}}$. We proceed with a proposition
that demonstrates how we can ``convert'' a $K-$approximation set for $s_{i}^{\mbox{inc}}$, to a $K-$approximation
set for $s_{i}.$
\begin{defn}
\label{def:Phi-inc}Let $\varphi:\left\{ A,...,B\right\} \rightarrow\mathbb{Z}^{+}$ be a nondecreasing function.

Let $\mbox{StirctIn\ensuremath{c_{\varphi}}}=\left\{ A,B\right\} \bigcup\left\{ i\in\mathbb{Z}|\ensuremath{A+1\leq i\leq B}\wedge\varphi\left(i\right)>\varphi\left(i-1\right)\right\} $
be the set of points where $\varphi$ is strictly increasing, and the edge points. Let $\mbox{Inc}_{\varphi}$ be
an arbitrary set that contains $\mbox{StrictInc}_{\varphi},$ i.e. $\mbox{StrictInc}_{\varphi}\subseteq\mbox{Inc}_{\varphi}\subseteq\left\{ A,...,B\right\} .$

Let us denote its elements by $A=k_{1}<k_{2}<...<k_{\left|\mbox{Inc}_{\varphi}\right|}=B$.

Let us define $\varphi^{\mbox{dom}}:\left\{ 1,...,|\mbox{In\ensuremath{c_{\varphi}}}|\right\} \rightarrow\left\{ A,...,B\right\} $
by $\varphi^{\mbox{dom}}\left(j\right)=k_{j}$, i.e. $\varphi^{\mbox{dom}}\left(j\right)$ is the $j$'th smallest
element in $\mbox{Inc}_{\varphi}.$ Note that this element belongs to the domain of $\varphi$, i.e. to $\left\{ A,...,B\right\} .$

Let us define $\varphi^{\mbox{inc}}:\left\{ 1,...,|\mbox{In\ensuremath{c_{\varphi}}}|\right\} \rightarrow\mathbb{Z}^{+}$
by $\varphi^{\mbox{inc}}\left(j\right)=\varphi\left(\varphi^{\mbox{dom}}\left(j\right)\right),$ i.e. the value
of $\varphi$ on the $j$'th smallest element of $\mbox{In\ensuremath{c_{\varphi}}}$. \end{defn}
\begin{rem}
Formally, as $\varphi^{\mbox{dom}}$ and $\varphi^{\mbox{inc}}$ are both based on the function $\varphi$ and the
set $\mbox{Inc}_{\varphi}$, we should have denoted them by $\varphi_{\mbox{Inc}_{\varphi}}^{\mbox{dom}}$ and $\varphi_{\mbox{Inc}_{\varphi}}^{\mbox{dom}}$
respectively. For simplicity, we omit the dependence on the set $\mbox{Inc}_{\varphi}.$\end{rem}
\begin{defn}
Let $W^{\mbox{inc}}=\left\{ 1=w_{1}<...<w_{r}=|\mbox{Inc}_{\varphi}|\right\} $ be a $K-$approximation set of $\varphi^{\mbox{inc}}$.

Let $\mbox{dom}\left(W^{\mbox{inc}}\right)=\left\{ \varphi^{\mbox{dom}}\left(w_{i}\right)|w_{i}\in W^{\mbox{inc}}\right\} $
be the elements in $\left\{ A,...,B\right\} $ referred to in $W^{\mbox{inc}}.$ Let us denote them by $\left\{ j_{1}<...<j_{r}\right\} $
where $j_{i}=\varphi^{\mbox{dom}}\left(w_{i}\right)$.

For every subset $S=\left\{ s_{1},...,s_{|S|}\right\} \subseteq\left\{ A,...,B\right\} $ let 
\[
\mbox{pad}\left(S\right)=\left\{ s_{i},s_{i}-1|2\leq i\leq|S|\right\} \bigcup\left\{ s_{1}\right\} 
\]
 consist of the set $S,$ and every of its elements padded with its previous element in $\left\{ A,...,B\right\} .$\end{defn}
\begin{prop}
\label{prop:inc K apx}Let $K\geq1$, and let $\varphi:\left\{ A,...,B\right\} \rightarrow\mathbb{Z}^{+}$ be a
nondecreasing function. Let $W^{\mbox{inc}}=\left\{ w_{1},...,w_{r}\right\} $ be a $K-$approximation set of $\varphi^{\mbox{inc}}.$
Let $W=\mbox{pad}\left(\mbox{dom}\left(W^{\mbox{inc}}\right)\right)$. Then $W$ is a $K-$approximation set of
$\varphi$.\end{prop}
\begin{proof}
Suppose $W=\left\{ A=a_{1}<a_{2}<...<a_{\ell}=B\right\} .$ In order to prove that $W$ is a $K-$approximation
set of $\varphi,$ we need to show that $\varphi\left(a_{j+1}\right)\leq K\varphi\left(a_{j}\right)$ for each $j=1,...,\ell-1$
that satisfies $a_{j+1}-a_{j}>1$.

Suppose $a_{j+1}-a_{j}>1$. There exists $k\in\left\{ 1,...,r-1\right\} $ such that $a_{j+1}=\varphi^{\mbox{dom}}\left(w_{k+1}\right)-1$,
and $a_{j+2}=\varphi^{\mbox{dom}}\left(w_{k+1}\right).$ 

If $w_{k+1}-w_{k}>1,$ then $\varphi^{\mbox{inc}}\left(w_{k+1}\right)\leq K\varphi^{\mbox{inc}}\left(w_{k}\right)$
because $W^{\mbox{inc}}$ is a $K-$approximation set of $\varphi^{\mbox{inc}}$. By definition of $\varphi^{\mbox{inc}},$
$\varphi\left(a_{j+2}\right)=\varphi^{\mbox{inc}}\left(w_{k+1}\right)\leq K\varphi^{\mbox{inc}}\left(w_{k}\right)=K\varphi\left(a_{j}\right)$,
and by the fact that $\varphi$ is nondecreasing, $\varphi\left(a_{j+1}\right)\leq\varphi\left(a_{j+2}\right)$.
That implies $\varphi\left(a_{j+1}\right)\leq K\varphi\left(a_{j}\right)$.

If $w_{k+1}-w_{k}=1$, then there are no increasing points between $a_{j}$ and $a_{j+2}$, which implies $\varphi\left(a_{j+1}\right)=\varphi\left(a_{j}\right)$,
and this indicates $\varphi\left(a_{j+1}\right)\leq K\varphi\left(a_{j}\right)$ .
\end{proof}
The next proposition is based on Proposition \ref{prop:4.6}.
\begin{prop}
\label{prop:knapsack}Let $\varphi:\left\{ A,...,B\right\} \rightarrow\mathbb{Z}^{+}$ be a nondecreasing function.
Let $\bar{\varphi}$ be a nondecreasing $L-$approximation function of $\varphi$ ($L$>1). Let $\mbox{Inc}_{\bar{\varphi}}$
a set which contains $\mbox{StrictInc}_{\bar{\varphi}},$ and denote its elements by $k_{1}=A<k_{2}<...<k_{m}=B.$
Let $W^{\mbox{inc}}$ be the output of function \textsc{ApxSet}\emph{ }for given parameters $\bar{\varphi}^{\mbox{inc}},\left\{ 0,...,m\right\} $,
and $K>1$. Let $W=\mbox{pad}\left(\mbox{dom}\left(W^{\mbox{inc}}\right)\right).$ Let $\hat{\bar{\varphi}}$ be
the approximation of $\bar{\varphi}$ induced by $W$. Let $t_{\bar{\varphi}},t_{\bar{\varphi}^{\mbox{dom}}}$ be
an upper bound on the time needed to evaluate $\bar{\varphi},\bar{\varphi}^{\mbox{dom}}$ respectively. Then $\hat{\bar{\varphi}}$
is a nondecreasing $KL-$approximation function of $\varphi$, the computation of $W$ takes $O\left(t_{\bar{\varphi}^{\mbox{dom}}}t_{\bar{\varphi}}\left(1+\log_{K}\varphi^{\max}\right)\log|\mbox{In\ensuremath{c_{\bar{\varphi}}}}|\right)$
time, and $|W|=O\left(1+\log_{K}\varphi^{\max}\right)$.\end{prop}
\begin{proof}
By Proposition \ref{prop:4.6}, $W^{\mbox{inc}}$ is a $K-$approximation set of $\bar{\varphi}^{\mbox{inc}}.$
By Proposition \ref{prop:inc K apx} $W$ is a $K-$approximation set of $\bar{\varphi}.$ By Proposition \ref{prop:4.5}
$\hat{\bar{\varphi}}$ is nondecreasing.

Applying Proposition \ref{prop:6.2} (approximation of approximation sets) with $\varphi_{1}=\bar{\varphi},\,\varphi_{2}=\varphi,\, K_{1}=K,\, K_{2}=L,$
and $W_{1}=W,$ we get that $\hat{\bar{\varphi}}$ is a $KL-$approximation of $\varphi.$

Running time: By Proposition \ref{prop:4.6}, the computation of $W^{\mbox{inc}}$ takes $O\left(t_{\bar{\varphi}}\left(1+\log_{K}\varphi^{\max}\right)\log m\right)$
time, and $|W^{\mbox{inc}}|=O\left(1+\log_{K}\varphi^{\max}\right)$. Since building $W$ when $W^{\mbox{inc}}$
is on hand takes $O\left(t_{\bar{\varphi}^{\mbox{dom}}}\left(1+\log_{K}\varphi^{\max}\right)\right)$, then the
computation time of $W$ and $|W|$ are the same as for $W^{\mbox{inc}}$ and $|W^{\mbox{inc}}|$. 
\end{proof}

\subsection{Algorithm}

We first express the dynamic programming formula of the problem:
\begin{eqnarray*}
s_{0}\left(j\right)= & 1 & j\geq0\\
s_{i}\left(j\right)= & 0 & j<0\\
s_{i}\left(j\right)= & s_{i-1}\left(j\right)+s_{i-1}\left(j-w_{i}\right) & j\geq0,\, i\geq1
\end{eqnarray*}

\textbf{Intuition to Algorithm 4:} In every iteration $i,$ we have from the former iteration: (i) the function
$\hat{\bar{s}}_{i-1},$ which is a succinct approximation of $s_{i-1},$ and (ii) a set $W_{i-1}+\left\{ 1\right\} $
(recall that $A+B=\left\{ a+b|\, a\in A,b\in B\right\} $) which is the set of $\hat{\bar{s}}_{i-1}$'s strictly
increasing points (and maybe one additional point).

Our goal is to construct a succinct $K-$approximation function for $s_{i}.$ Apparently, it could be done easily
by iteratively constructing $K-$approximation sets for the functions $\bar{s}_{i}\left(\cdot\right)=\hat{\bar{s}}_{i-1}\left(\cdot\right)+\hat{\bar{s}}_{i-1}\left(\cdot-w_{i}\right),\,0\leq i\leq n.$
The running time of this procedure is $O\left(\frac{n^{3}}{\epsilon}\log\frac{n}{\epsilon}\log C\right),$ see Appendix
\ref{sec:Not-strongly-polynomial}. This means that the running time of the algorithm depends on $C-$ the capacity
of the knapsack, and the algorithm is therefore not strongly polynomial. 

In order to overcome this problem, we first develop a $K-$approximation set for $\bar{s}_{i}^{\mbox{inc}},$ that
does not depend on the numbers of the problem (an upper bound on the range of $\bar{s}_{i}^{\mbox{inc}}$ is $2^{n}$).
Then we ``convert'' it to a $K-$approximation set of the function $s_{i},$ as demonstrated in Proposition \ref{prop:inc K apx}.
This way, we get a $K-$approximation set for $s_{i},$ and the running time does not depend on the numbers of the
problem, so the algorithm is strongly polynomial.\\
\\
We give a remark on step 4(b) of Algorithm 4: $\mbox{Inc}_{i}$ stores in a sorted array of the form $\left\{ \left(j,x_{j}\right)|1\leq j\leq\left|\mbox{Inc}_{i}\right|\,,\, x_{j}\in\mbox{Inc}_{i}\right\} $.

\begin{algorithm}[h]
\protect\caption{Counting 0\textbackslash{}1 knapsack}

\begin{enumerate}
\item \textbf{Function }\textbf{\textsc{StrongFPTASKnapsack}}
\item $K\leftarrow\sqrt[n]{1+\epsilon}$ , $W_{0}\leftarrow\left\{ 0,C\right\} $
\item let $\hat{\bar{s}}_{0}\left(j\right)=1$ for all $0\leq j\leq C$
\item for $i:=1\mbox{ to }n$

\begin{enumerate}
\item let $\bar{s}_{i}\left(\cdot\right)=\hat{\bar{s}}_{i-1}\left(\cdot\right)+\hat{\bar{s}}_{i-1}\left(\cdot-w_{i}\right)$
\item $\mbox{Inc}_{i}\leftarrow\left(\left(W{}_{i-1}+\left\{ 1\right\} \right)\bigcup\left(W{}_{i-1}+\left\{ w_{i}+1\right\} \right)\bigcup\left\{ 0,C\right\} \right)\bigcap\left\{ 0,\ldots,C\right\} $ 
\item $W_{i}^{\mbox{inc}}\leftarrow$\textsc{ApxSet}$\left(\bar{s}_{i}^{\mbox{inc}}\left(\cdot\right),\left\{ 1,\ldots,\left|\mbox{Inc}_{i}\right|\right\} ,K\right)$\textbackslash{}{*}
$\bar{s}_{i}^{\mbox{inc}}$ as defined in Definition \ref{def:Phi-inc}{*}\textbackslash{}
\item $W_{i}\leftarrow$$\mbox{pad}\left(\mbox{dom}\left(W_{i}^{\mbox{inc}}\right)\right)$
\item let $\hat{\bar{s}}_{i}\left(\cdot\right)$ be the approximation of $\bar{s}_{i}\left(\cdot\right)$ induced by $W_{i}$
\end{enumerate}
\item end for
\item \textbf{return} $\hat{\bar{s}}_{n}\left(C\right)$\end{enumerate}
\end{algorithm}

\begin{prop}
\label{prop:Knapsack strongly}Let $0<\epsilon<1$, and let $w_{1},...,w_{n},C$ be an instance of a knapsack problem.
Then $\hat{\bar{s}}_{n}\left(C\right)$ calculated by Algorithm 4 is a $\left(1+\epsilon\right)-$approximation
function of $s_{n}\left(C\right)$. The algorithm is deterministic and runs in time $O\left(\frac{n^{3}}{\epsilon}\log^{2}\frac{n}{\epsilon}\right)$.\end{prop}
\begin{proof}
\textbf{Correctness:} For every $i$ the function $\bar{s}_{i}^{\mbox{inc}}\left(\cdot\right)$ is nonnegative nondecreasing
as a sum of such functions. So the call to \textsc{ApxSet} is well defined.

Next, we prove that Algorithm 4 returns $\left(1+\epsilon\right)-$approximation solution. To do so, we show that
for every iteration $0\leq i\leq n$ the following 4 properties hold:

(1) $W{}_{i}$ is a $K-$approximation set of $\bar{s}_{i}\left(\cdot\right)$, and $W_{i}+\left\{ 1\right\} \supseteq\mbox{SrtictInc}_{\hat{\bar{s}}_{i}}$

(2) $\hat{\bar{s}}_{i}\left(\cdot\right)$ is a nondecreasing $K^{i}-$approximation function of $s_{i}$.

(3) $W_{i}^{\mbox{inc}}$ is a $K-$approximation set of $\bar{s}_{i}^{\mbox{inc}}\left(\cdot\right)$.

(4) $\mbox{Inc}_{i}\supseteq\mbox{StrictInc}_{\bar{s}_{i}\left(\cdot\right)}$.

We prove it by induction. The induction hypothesis is about properties (1) and (2) only.

The base case of $i=0$: $W_{0}$ is a $1-$approximation set of the function $s_{0}\left(j\right)\equiv1$, and
$\hat{\bar{s}}_{0}\left(\cdot\right)$ is the $1-$approximation of $s_{0}\left(\cdot\right)$ induced by $W_{0}$.

Our induction hypothesis is that properties (1)+(2)are valid for $i-1$. We now prove that all 4 properties are
valid for $i$.

We start with property (4). First, note that for any arbitrary nondecreasing step function $\varphi,$ by definition
of the function $\hat{\varphi}$ induced by a $K-$approximation set $\bar{W}$, $\bar{W}+\left\{ 1\right\} $ contains
the points where $\hat{\varphi}$ is strictly increasing. Therefore $W_{i-1}+\left\{ 1\right\} $ contains the points
where $\hat{\bar{s}}_{i-1}\left(\cdot\right)$ is strictly increasing, and $W_{i-1}+\left\{ w_{i}+1\right\} $ contain
the points where $\hat{\bar{s}}_{i-1}\left(\cdot-w_{i}\right)$ is strictly increasing. Second, note that for any
arbitrary step functions $\varphi$ and $\psi,$ $\varphi+\psi$ is strictly increasing in $i$, if and only if
$\varphi$ is strictly increasing in $i$ or $\psi$ is strictly increasing in $i.$ Thus, by Property (1) of the
induction hypothesis, $\mbox{Inc}_{i}$ includes the points where $\bar{s}_{i}$ is strictly increasing.

We proceed by proving the other 3 properties: property (3) is valid due to Proposition \ref{prop:4.6}  with parameters
set to $\varphi=\bar{s}^{\mbox{inc}},\, D=\left\{ 1,\ldots,\left|\mbox{Inc}_{i}\right|\right\} $ and $K=K.$ Property
(1) is valid by Proposition \ref{prop:inc K apx} with parameters set to $\varphi=s_{i}\left(\cdot\right),\, K=K\,,\, W^{\mbox{inc}}=W_{i}^{\mbox{inc}}$.
Property (2) derived from Proposition \ref{prop:knapsack} with parameters set to $\varphi=s_{i}\left(\cdot\right),\,\bar{\varphi}=\bar{s}_{i}\left(\cdot\right)\,,\, K=K\,,\, L=K^{i-1}$
(by Proposition \ref{prop:5.1}(1) (summation of approximation) and the dynamic formula, $\bar{s}_{i}\left(\cdot\right)$
is a $K^{i-1}-$approximation of $s_{i}\left(\cdot\right)$, using Property (2) of the induction hypothesis). This
completes the proof by induction.

Recall that $K=\sqrt[n]{1+\epsilon}$. We deduce from property (1) above with $i=n$, that for every $0\leq j\leq C$
we have $s_{n}\left(j\right)\leq\hat{\bar{s}}_{n}\left(j\right)\leq\left(\left(\sqrt[n]{1+\epsilon}\right)^{n}\right)s_{n}\left(j\right).$
When $j=C$, we therefore have $s_{n}\left(C\right)\leq\hat{\bar{s}}_{n}\left(C\right)\leq\left(1+\epsilon\right)s_{n}\left(C\right)$.
That proves the approximation ratio.

\textbf{Running time:} Clearly, the running time of the algorithm is dominated by the for-loop, which has $n$ iterations.
We first show that the running time of each iteration is dominated by step 4(c). In 4(b) we merge $W{}_{i-1}^{\mbox{}}+\left\{ 1\right\} $
and $W{}_{i-1}^{\mbox{}}+\left\{ w_{i}+1\right\} $ to a sorted array. Since the set $W_{i-1}$ is a $K-$approximation
set (of some function), then by definition it is sorted, and thus $W_{i-1}+\left\{ w_{i}\right\} $ is sorted too.
The cardinality of each is $O\left(|W{}_{i-1}^{\mbox{}}|\right),$ and by Proposition \ref{prop:knapsack}, $O\left(|W{}_{i-1}^{\mbox{}}|\right)=O\left(1+\log_{K}2^{n}\right)$,
so the merge operation takes $O\left(1+\log_{K}2^{n}\right)$ time. Step 4(d) runs in $O\left(|W{}_{i}^{\mbox{inc}}|\right)=O\left(1+\log_{K}2^{n}\right)$:
since $\mbox{Inc}_{i}$ is stored in a sorted array of the form $\left\{ \left(j,x_{j}\right)|1\leq j\leq\left|\mbox{Inc}_{i}\right|\,,\, x_{j}\in\mbox{Inc}_{i}\right\} $
(step 4(b)), the operation of $\varphi^{\mbox{dom}}\left(w\right)$ takes $O\left(1\right)$ time for every $w\in W_{i}^{\mbox{inc}}$.
Therefore the running time of each iteration is dominated by the call to \textsc{ApxSet} in 4(c). 

Since $O\left(|\mbox{Inc}_{i}|\right)=O\left(1+\log_{K}2^{n}\right)$ by Proposition \ref{prop:knapsack} , the
running time of \textsc{ApxSet} is $O\left(t_{\bar{s}_{i}\left(\cdot\right)}\log_{K}2^{n}\log\log_{K}2^{n}\right)$
($2^{n}$ is an upper bound for $s_{n}\left(\cdot\right)$). $\hat{\bar{s}}_{i-1}\left(\cdot\right)$ is a function
induced by a $K-$approximation set, and by Proposition \ref{prop:knapsack} it can be saved efficiently, so $t_{\bar{s}_{i}\left(\cdot\right)}=O\left(\log\log_{K}2^{n}\right)$.
We can rely on the fact $O\left(\log_{K}2^{n}\right)=O\left(\frac{n\log2^{n}}{\log\left(1+\epsilon\right)}\right)=O\left(\frac{n^{2}}{\epsilon}\right)$,
which holds true for every $0\leq\epsilon\leq1,$ and thereby conclude that the running time is $O\left(\frac{n^{3}}{\epsilon}\log^{2}\frac{n}{\epsilon}\right)$. 
\end{proof}
The resultant running time is, up to log term, similar to \cite{key-6}'s running time $O\left(\frac{n^{3}}{\epsilon}\log\frac{n}{\epsilon}\right)$.
Both running times are strongly polynomial. \cite{key-6}'s method is specific to the 0/1 knapsack problem, whereas
our's is a general method.

\section{Strongly polynomial algorithm for counting $m-$tuples \label{sec:Strongly-polynomial-m-tuples}}

The DP formulation for solving the problem of counting $m-$tuples have similar structure as the one of counting
0/1 knapsack. In addition, the solutions of both problems have an upper bound that is independent on the number
of the problem (in knapsack- $2^{n}$, and in $m-$tuples- $\prod_{i=1}^{m}\left|X_{i}\right|$). These similar
properties of both problems enable us to use the method developed in Section \ref{sec:0/1-knapsack--strongly} and
apply it to the problem of counting $m-$tuples.

\begin{algorithm}[h]
\protect\caption{Counting $m-$tuples by a strongly polynomial algorithm}

\begin{enumerate}
\item \textbf{Function }\textbf{\textsc{StrongFPTASMtuples}}
\item $K\leftarrow\sqrt[m]{1+\epsilon}$ , $\mbox{Inc}_{1}=\left\{ 0,x_{11},\ldots,x_{1\left|X_{1}\right|},B\right\} $ 
\item $W_{1}^{\mbox{Inc}}\leftarrow$\textsc{ApxSet}$\left(z_{1}^{\mbox{inc}}\left(\cdot\right),\left\{ 1,\ldots,\left|\mbox{Inc}_{1}\right|\right\} ,K\right)$,
$W_{1}\leftarrow$$\mbox{pad}\left(\mbox{dom}\left(W_{1}^{\mbox{inc}}\right)\right)$
\item let $\hat{\bar{z}}_{1}\left(j\right)$be the approximation of $\bar{z}_{1}\left(\cdot\right)$ induced by $W_{1}$
\item for $i:=2\mbox{ to }m$

\begin{enumerate}
\item let $\bar{z}_{i}\left(j\right)=\sum_{k=1}^{|X_{i}|}\hat{\bar{z}}_{i-1}\left(j-x_{ik}\right)$
\item $\mbox{Inc}_{i}\leftarrow\left(\bigcup_{j=1}^{\left|X_{i}\right|}\left(W{}_{i-1}+\left\{ x_{ij}+1\right\} \right)\bigcup\left\{ 0,C\right\} \right)\bigcap\left\{ 0,\ldots,C\right\} $
as a sorted array
\item $W_{i}^{\mbox{inc}}\leftarrow$\textsc{ApxSet}$\left(\bar{z}_{i}^{\mbox{inc}}\left(\cdot\right),\left\{ 1,\ldots,\left|\mbox{Inc}_{i}\right|\right\} ,K\right)$\textbackslash{}{*}
$\bar{z}_{i}^{\mbox{inc}}$ as defined in Definition \ref{def:Phi-inc}{*}\textbackslash{}
\item $W_{i}\leftarrow$$\mbox{pad}\left(\mbox{dom}\left(W_{i}^{\mbox{inc}}\right)\right)$
\item let $\hat{\bar{z}}_{i}\left(\cdot\right)$ be the approximation of $\bar{z}_{i}\left(\cdot\right)$ induced by $W_{i}$
\end{enumerate}
\item end for
\item \textbf{return} $\hat{\bar{z}}_{m}\left(B\right)$\end{enumerate}
\end{algorithm}

\begin{prop}
Let $0<\epsilon<1$, and let $X_{1},\ldots,X_{m},B$ be an instance of a $m-$tuples problem. Then $\hat{\bar{z}}_{m}\left(B\right)$
calculated by Algorithm 5 is a $\left(1+\epsilon\right)-$approximation function of $z_{m}\left(B\right)$. The
algorithm is deterministic and runs in time 
\[
O\left(\frac{m^{3}}{\epsilon}\left(\sum_{i=1}^{m}\left|X_{i}\right|\right)\log\left(\prod_{i=1}^{m}\left|X_{i}\right|\right)\log\frac{m\log\prod_{i=1}^{m}\left|X_{i}\right|}{\epsilon}\log\left(\frac{m\prod_{i=1}^{m}\left|X_{i}\right|\log\prod_{i=1}^{m}\left|X_{i}\right|}{\epsilon}\right)\right)
\]
.\end{prop}
\begin{proof}
The proof of correctness and the approximation ratio is very similar to the proof of Proposition \ref{prop:Knapsack strongly}.
The changes are minor, and therefore we omit this part of the proof.

\textbf{Running time:} Clearly, the running time of the algorithm is dominated by the for-loop, which has $m-1$
iterations (step3 runs in similar time to the other iterations in the for-loop) . We first show that the running
time of each iteration is dominated by step 5(c). In 5(b) we merge $W{}_{i-1}^{\mbox{}}+\left\{ 1\right\} .\ldots,W_{i-1}+\left\{ x_{1\left|X_{1}\right|}+1\right\} $
to a sorted array. Since the set $W_{i-1}$ is a $K-$approximation set (of some function), then by definition it
is sorted, and thus $W_{i-1}+\left\{ x_{ij}\right\} $ is sorted for any $j$. The cardinality of each is $O\left(|W{}_{i-1}^{\mbox{}}|\right),$
and by Proposition \ref{prop:knapsack}, $O\left(|W{}_{i-1}^{\mbox{}}|\right)=O\left(1+\log_{K}\prod_{i=1}^{m}\left|X_{i}\right|\right)$,
so the merge operation takes $O\left(1+\log_{K}\prod_{i=1}^{m}\left|X_{i}\right|\right)$. Since step 5(d) runs
in $O\left(|W{}_{i}^{\mbox{inc}}|\right)=O\left(1+\log_{K}\prod_{i=1}^{m}\left|X_{i}\right|\right)$, the running
time of each iteration is dominated by the call to \textsc{ApxSet} in 5(c). 

Since $O\left(|\mbox{Inc}_{i}|\right)=O\left(\left|X_{i}\right|\log_{K}\prod_{i=1}^{m}\left|X_{i}\right|\right)$
by Proposition \ref{prop:knapsack} , the running time of \textsc{ApxSet} in the $i$'th iteration is $O\left(t_{\bar{z}_{i}\left(\cdot\right)}\log_{K}\prod_{i=1}^{m}\left|X_{i}\right|\log\left(\left|X_{i}\right|\log_{K}\prod_{i=1}^{m}\left|X_{i}\right|\right)\right)$
($\prod_{i=1}^{m}\left|X_{i}\right|$ is an upper bound for $z_{m}\left(\cdot\right)$). $\hat{\bar{z}}_{i-1}\left(\cdot\right)$
is a function induced by a $K-$approximation set, and by Proposition \ref{prop:knapsack} it can be stored efficiently,
so $t_{\hat{\bar{z}}_{i}\left(\cdot\right)}=O\left(\log\log_{K}\prod_{i=1}^{m}\left|X_{i}\right|\right)$ thus by
definition $t_{\bar{z}_{i}\left(\cdot\right)}=O\left(\left|X_{i}\right|\log\log_{K}\prod_{i=1}^{m}\left|X_{i}\right|\right)$
. We can rely on the fact $O\left(\log_{K}\prod_{i=1}^{m}\left|X_{i}\right|\right)=O\left(\frac{m\log\prod_{i=1}^{m}\left|X_{i}\right|}{\log\left(1+\epsilon\right)}\right)=O\left(\frac{m\log\prod_{i=1}^{m}\left|X_{i}\right|}{\epsilon}\right)$,
which holds true for every $0\leq\epsilon\leq1,$ and thereby conclude that the running time is 
\begin{eqnarray*}
O\left(\frac{m^{2}}{\epsilon}\left(\sum_{i=1}^{m}\left|X_{i}\right|\right)\log\left(\prod_{i=1}^{m}\left|X_{i}\right|\right)\log\frac{m\log\prod_{i=1}^{m}\left|X_{i}\right|}{\epsilon}\sum_{i=1}^{m}\log\left(\frac{m\left|X_{i}\right|\log\prod_{i=1}^{m}\left|X_{i}\right|}{\epsilon}\right)\right)\\
=O\left(\frac{m^{3}}{\epsilon}\left(\sum_{i=1}^{m}\left|X_{i}\right|\right)\log\left(\prod_{i=1}^{m}\left|X_{i}\right|\right)\log\frac{m\log\prod_{i=1}^{m}\left|X_{i}\right|}{\epsilon}\log\left(\frac{m\left(\prod_{i=1}^{m}\left|X_{i}\right|\right)\log\prod_{i=1}^{m}\left|X_{i}\right|}{\epsilon}\right)\right)
\end{eqnarray*}
. 
\end{proof}
Comparing to the not strongly polynomial algorithm in Algorithm 2, for achieving the property of strongly polynomial
running time, wetrade off the term $\log B$ with the term $m\log\left(\frac{m\left(\prod_{i=1}^{m}\left|X_{i}\right|\right)\log\prod_{i=1}^{m}\left|X_{i}\right|}{\epsilon}\right)$
in the running time of the algorithm.

\appendix

\section{ApxSet for nonincreasing function\label{sec:-ApxSet-for}}

\begin{algorithm}[h]
\protect\caption{\label{alg:ApxSet-1}Constructing a $K-$approximation set for a nonincreasing function $\varphi$}

\begin{enumerate}
\item \textbf{Function }\textbf{\textsc{ApxSet}}$\left(\varphi,\left\{ A,\ldots,B\right\} ,K\right)$
\item $x\leftarrow A$
\item $W\leftarrow\left\{ A,B\right\} $
\item while $x<B$ do

\begin{enumerate}
\item $x\leftarrow\min\left\{ x+1,\min\left\{ y\in\left\{ A,\ldots,B\right\} |K\varphi\left(y\right)\geq\varphi\left(x\right)\right\} \right\} $
\item $W\leftarrow W\bigcup\left\{ x\right\} $
\end{enumerate}
\item end while
\item \textbf{return} $W$\end{enumerate}
\end{algorithm}

\section{Not strongly polynomial algorithm for counting 0/1 knapsack solutions\label{sec:Not-strongly-polynomial}}

\begin{algorithm}[h]
\protect\caption{The not strongly polynomial algorithm for counting 0\textbackslash{}1 knapsack solutions}

\begin{enumerate}
\item \textbf{Function}\textbf{\textsc{ FPTASKnapsack}}
\item K$\leftarrow$$\sqrt[{\normalcolor n}]{{\normalcolor 1+\epsilon}}$, $W_{0}\leftarrow\left\{ 0,C\right\} $
\item let $\hat{\bar{s}}_{0}\left(j\right)=1$ for all $0\leq j\leq C$
\item for i:=1 to $n$ 

\begin{enumerate}
\item let $\bar{s}_{i}\left(\cdot\right)=\hat{\bar{s}}_{i-1}\left(\cdot\right)+\hat{\bar{s}}_{i-1}\left(\cdot-w_{i}\right)$
\item $W_{i}\leftarrow\mathbf{ApxSet}$$\left(\bar{s}_{i}\left(\cdot\right),\left\{ 0,...,C\right\} ,K\right)$
\item let $\hat{\bar{s}}_{i}\left(\cdot\right)$ be the approximation of $\bar{s}_{i}\left(\cdot\right)$ induced by $W_{i}$
\end{enumerate}
\item end for
\item \textbf{return} $\hat{\bar{s}}_{n}\left(C\right)$\end{enumerate}
\end{algorithm}

We show that algorithm 7 is a not strongly polynomial algorithm to approximate the number of 0/1 knapsack solutions. 

The correctness of the algorithm and the approximation ratio follow by arguments similar to those of Algorithm 2.
We next analyze the running time of the algorithm: Clearly, the running time of the algorithm is dominated by the
for-loop, which has $n$ iterations. In each iteration, the running time is dominated by the execution of function
\textsc{ApxSet} in step 4(b). By Proposition \ref{prop:4.6}, the running time is $O\left(t_{\bar{s}_{i}\left(\cdot\right)}\log_{K}2^{n}\log C\right)$
(Note that $2^{n}$ is an upper bound on the number of knapsack solutions). As a function induced by a $K-$approximation
set, $\tilde{s}_{i-1}\left(\cdot\right)$ can be stored efficiently, so by Proposition \ref{prop:4.6} $t_{\bar{s}_{i}\left(\cdot\right)}=\log_{K}2^{n}$.
So the running time is $O\left(\log_{K}2^{n}\log\log_{K}2^{n}\log C\right)$. We can rely on the fact $O\left(\log_{K}2^{n}\right)=O\left(\frac{n\log2^{n}}{\log\left(1+\epsilon\right)}\right)=O\left(\frac{n^{2}}{\epsilon}\right)$,
which holds true for every $0\leq\epsilon\leq1,$ and thereby conclude that the running time is $O\left(\frac{n^{3}}{\epsilon}\log\frac{n}{\epsilon}\log C\right)$.

\end{document}